\documentclass[letterpaper, 10 pt, conference]{ieeeconf}  

\IEEEoverridecommandlockouts                              

\usepackage{amssymb}  
\usepackage{xcolor}
\usepackage{mathtools}
\usepackage{mathrsfs}  
\usepackage{acronym}
\usepackage{amsmath}
\usepackage{dsfont}
\usepackage{paralist}
\usepackage{algorithm}
\usepackage{algorithmic}
\usepackage{multirow}
\usepackage{booktabs}
\usepackage{array} 
\usepackage{bm}
\usepackage[T1]{fontenc}
\usepackage{subcaption}
\usepackage{hyperref}

  \newcommand{\ie}{\textit{i.e.}}

  \newcommand{\dist}{\mathcal{D} }

\newcommand{\reals}{\mathbf{R}}
\newtheorem{theorem}{Theorem} 
\newtheorem{definition}{Definition}

\newtheorem{remark}{Remark}

\newtheorem{assumption}{Assumption}

  \acrodef{mdp}[MDP]{Markov decision process} 
  \acrodefplural{mdp}[MDPs]{Markov decision processes}

 \newcommand{\abs}[1]{\lvert #1 \rvert}
 
 \DeclareMathOperator*{\optmax}{\textrm{maximize}}

 \DeclareMathOperator*{\optst}{\textrm{subject to}}

  \newcommand{\Expect}{\mathds{E}}

\newcommand{\calO}{\mathcal{O}}

\newcommand{\calK}{\mathcal{K}}

 \acrodef{ltl}[LTL]{Linear Temporal Logic}
  \acrodef{ltlf}[LTLf]{Linear Temporal Logic over Finite Traces}
 \acrodef{dfa}[DFA]{Deterministic Finite Automaton}
 \acrodef{pdfa}[PDFA]{Preference DFA}

\acrodef{mdp}[MDP]{Markov decision process}

\acrodef{asw}[ASW]{Almost-Sure Winning}

   \newcommand{\calM}{\mathcal{M}}

\title{\LARGE \bf
Planning Stealthy Backdoor Attacks  in MDPs with Observation-Based Triggers
}

\author{Xinyi Wei$^{1}$, Shuo Han$^{2}$, Ahmed Hemida$^{3}$, Charles A kamhoua$^{3}$ and Jie Fu$^{1}$
\thanks{$^{1}$Xinyi Wei and Jie Fu  are with the Department of Electrical and Computer Engineering,
        University of Florida, Gainesville, FL 32611, USA 
        {\tt\small weixinyi, fujie@ufl.edu}}%
\thanks{$^{2}$Shuo Han is with the Department of Electrical and Computer Engineering, University of Illinois Chicago, Chicago, IL 60607, USA
        {\tt\small hanshuo@uic.edu}}%
\thanks{$^{3}$Ahmed Hemida and Charles A kamhoua  are with DEVCOM Army Research Lab
        {\tt\small ahmed.h.hemida.ctr, 
    charles.a.kamhoua.civ@army.mil }}
    \thanks{Distribution A: Approved for public release; Distribution is unlimited.}%
}

\begin{document}

\maketitle
\thispagestyle{empty}
\pagestyle{empty}

\begin{abstract}
This paper investigates backdoor attack planning in stochastic control systems modeled as Markov Decision Processes (MDPs). A backdoor attack involves an adversary deploying a policy that performs well in the original MDP to pass testing, but behaves maliciously at runtime when combined with a trigger that perturbs system dynamics. We consider a sophisticated attacker capable of jointly optimizing the backdoor policy and its trigger using only a blackbox simulator. During execution, the attacker has access only to partial observations of the system state and is restricted to introduce small perturbations to the system's transition dynamics.  We formulate the attack planning problem as a constrained Markov game with an augmented state space and two players: Player 0 learns a backdoor policy that maximizes attack rewards when the trigger is active. However, when the trigger is inactive, the backdoor policy behaves near-optimally in the original MDP; Player 1 designs a finite-memory, observation-based trigger to activate the attack. We propose a switching gradient-based optimization algorithm to jointly solve for the backdoor policy and trigger. Experiments on a case study demonstrate the effectiveness of our method in achieving stealthy and successful backdoor attacks, and how the attack performance varies under different parameters related to the stealthiness of the backdoor attack.

\end{abstract}

\section{Introduction}

 The threat of backdoor attacks in \acp{mdp} and reinforcement learning (RL) has becomes an increasingly active topic \cite{kiourtiTrojDRLEvaluationBackdoor2020,gongBAFFLEHidingBackdoors2024,cuiBadRLSparseTargeted2024}. 
A backdoor attack in RL/MDP consists of a pair $(\pi^\dagger, \kappa)$  where $\pi^\dagger$ is a backdoor policy that achieves near-optimal performance in the system, and $\kappa$ is a trigger that inject noise into the system dynamics or the observation such that the policy $\pi^\dagger$ under trigger-perturbed dynamics or observation can introduce significant performance degradation. Similar to backdoor attacks in supervised learning \cite{goldblumDatasetSecurityMachine2023,goldwasserPlantingUndetectableBackdoors2022,liDeepLearningBackdoors2022}, RL backdoor attacks leverage data poisoning attacks, where the adversary compromises certain training samples to induce the learned backdoor policy, and a trigger that perturbs the system state/observation in runtime with a small noise to activate the backdoor.

This paper investigates a class of backdoor attacks in Markov Decision Processes (MDPs), where an adversary provides a policy that performs well under normal conditions but exhibits malicious behavior when activated  during online execution. In particular, the attacker has partial observation over the system states and can introduce small, hard-to-detect perturbations to the system’s transition dynamics using a finite-memory, observation-based trigger.  For example, an attacker could alter the environment by tampering with network routing conditions or slightly de-synchronizing a motor. The attack planning problem is to jointly design the backdoor policy and its trigger so that the policy optimizes the attacker’s objectives when the trigger is activated. We begin with a model-based analysis to gain insight into such coordinated backdoor attacks, and then show that the attacker can compute an optimal strategy using only a blackbox simulator of the original system and its observation function.

Specifically, for the trigger to be stealthy, we restrict the attacker's actions to allow perturbing  the system's transition function within a $d$-rectangular uncertainty set \cite{nilimRobustControlMarkov2005,wiesemannRobustMarkovDecision2013} of the original MDP. By restricting perturbations to a finite set of transition functions within this set, we formulate the attack planning using an augmented constrained Markov game whose state space includes both the system states and the trigger's memory states, and the joint action space comprises the original MDP actions and  a finite set of transition functions, including both perturbed  and   original transition functions. We show that, within a constrained policy space, a constrained equilibrium in this augmented game yields an optimal backdoor strategy $(\pi^\dagger, \kappa)$ such that:  
\begin{inparaenum}[1)]
    \item $\pi^\dagger$ remains near-optimal under the original (unperturbed) dynamics; and  
    \item when the trigger $\kappa$ is activated, the induced perturbations maximize the attacker’s total discounted return.
\end{inparaenum}
We then provide a  gradient based algorithm to solve for the coordinated  optimal backdoor strategy $(\pi^\dagger, \kappa)$ using simulated trajectories from a blackbox simulator and the attacker's observation function.

 \subsection{Related Work and Our Contribution} 
TrojDRL \cite{kiourtiTrojDRLEvaluationBackdoor2020}   investigated poisoning attacks in DRL combined with runtime trigger which can alter the state, action, or reward information to carry out targeted attacks (ensuring the trained policy takes a certain targeted action at a certain state), or untargeted attacks which make the agent to uniformly select actions from a fraction of states. In comparison, we consider a case when the attacker does not aim for a targeted policy or uniform random action choices, but rather jointly design the backdoor policy and trigger to maximize an attacker's reward under stealthy constraints. Our synthesized backdoor policy can be used to generate targeted attacks in TrojDRL. 
BadRL  \cite{cuiBadRLSparseTargeted2024}  builds on TrojRL to enforce sparsity in the targeted attack. It  introduces  a heuristic approach to determine attacks on states with high values.  Our proposed method also assumes the attacker has access to a simulator. Instead of determining heuristically attack states with high-values, the attacker optimizes the coordinated attacks using a backdoor policy (trusted by the victim) and the trigger that introduces perturbations at runtime. 

In
\cite{gongBAFFLEHidingBackdoors2024}, the authors study a   backdoor attack for offline RL agent. The method first trains a weak policy using the offline data to minimize the total reward. It assumes a trigger that can perturb the observation of states $s$ with small noise $s'=s+\delta$, and train a poisoned policy that minimizes jointly the  difference between the poisoned policy and a normal policy, and the difference between the poisoned policy and the weak policy. In comparison, our work considers a partially observable attacker  that determines how to strategically perturb the system at runtime given imperfect observations. Our solutions are based on model-based analysis and constrained Markov games instead of minimizing policy differences.

Another line of work  \cite{bhartiProvableDefenseBackdoor2022a,gaoCooperativeBackdoorAttack2024} employs the policy $\pi^\dagger$ that has two key properties: 1) on a safe subspace that could be visited without perturbation, $\pi^\dagger$ and the original policy are the same; 2) on the unsafe subspace that cannot be visited under the original optimal policy, the methods design an adversarial policy to optimize the attack objective. During online execution,  a so-called \emph{subspace trigger} is designed to poison the observation so that the system can visit  the unsafe subspace. This class of backdoor attack has been initially studied for single-agent RL \cite{bhartiProvableDefenseBackdoor2022a} and then  extended to    decentralized RL~\cite{gaoCooperativeBackdoorAttack2024}, for    both backdoor attack and defense design. Our approach restricts the backdoor policy $\pi^\dagger$ to be $(1-\varepsilon)$-optimal, without requiring it to match the original policy exactly. This relaxation increases the attacker’s flexibility in crafting more effective exploits. We further demonstrate through experiments how the near-optimality parameter $\varepsilon$ influences attack performance.


\section{Preliminaries and Problem Formulation}
\label{sec:problem}
We begin by introducing the system model and threat model that form the basis for our problem formulation and analysis.
We use $\dist{(X)}$ to denote the probability simplex over the set $X$. Consider a probabilistic decision-making problem modeled as  an \ac{mdp} $M=( S, A,P, \mu_0 ,r)$, where the set $ S  $ is a finite state space, $A$ is a finite action space, $\mu_0$ is the initial state distribution. The function $P:S \times A\rightarrow \dist{(S)}$ represents the probabilistic transition kernel, and  $r:S\times A\rightarrow \reals$ is the reward function.

 Given a Markov  policy $\pi: S \rightarrow \dist{(A)}$ and the initial state distribution $\mu_0 \in \dist{(S)}$, the total discounted reward received by the planning agent  is \[ 
 V_0( \mu_0,M^\pi) = \Expect_\pi \left[
 \sum_{t=0}^\infty \gamma^t r (S_t,A_t)\mid S_0\sim \mu_0
 \right],
 \]  where $\gamma$ is a discount factor, and $S_t$, $A_t$ denote the state and action random variables at time step $t$, respectively, in the policy $\pi$-induced Markov chain $M^\pi$.

\subsection{Threat Model of Targeted Backdoor Attack}

In a backdoor attack, an adversary constructs a pair \((\pi^\dagger, \kappa)\), where \(\pi^\dagger\) denotes a backdoor policy and \(\kappa\) denotes a trigger. The adversary   releases \(\pi^\dagger\) for public use. Such a policy \(\pi^\dagger\) behaves close to an optimal policy without trigger; however, when the trigger \(\kappa\) is activated, its performance could significantly degrade. A victim, aiming to perform well in this \ac{mdp}, may unknowingly adopt \(\pi^\dagger\), unaware that the adversary can invoke the trigger to gain an advantage or cause poor behavior in an adversarial setting. 

\subsubsection{Trigger Design with Adversarially Perturbation}
We consider a class of triggers in which the attacker can alter the system dynamics by modifying the transition kernel of the underlying Markov decision process, thereby affecting the evolution of the state under a given policy.

 We introduce a finite set of \emph{adversarially perturbed} transition functions. Let \( K \) denote the number of perturbed modes. For each \( 1 \leq k \leq K \), let \( P_k \) represent a perturbed transition function that shares the same support as the original transition function \( P \). That is, for any state-action pair \( (s, a) \in S \times A \), a successor state \( s' \) is reachable under \( P_k \) only if it is also reachable under the original transition function \( P \).
 This class of triggers can be achieved by modifying the system actuator signals, or changing the  external  environment to induce a distribution shift.

  We denote the original (unperturbed) transition function as \( P_0 \coloneqq P \).   Define the index set \( \mathcal{K} = \{0, 1, \ldots, K\} \), which includes both the nominal dynamics and the \( K \) perturbed variants. The set $\mathcal{K}$ is understood as the trigger actions: Action $0$ means there is no adversarial perturbation, and action $k$, for $k \in \mathcal{K} \setminus \{0\}$ means that the attacker perturbs the dynamics such that the transition is governed by $P_k$.
   \subsubsection{Attacker's Partial Observations} 
We assume that the attacker has partial observability of the system states at runtime, modeled by an emission function \(E: S \rightarrow \dist(\mathcal{O})\), where \(\mathcal{O}\) is a finite observation set and \(E(o \mid s)\) gives the probability of observing \(o\) given latent state \(s\).



With the trigger actions and the nominal system model,   a class of backdoor strategies is defined as follows. 

\begin{definition}
\label{def:backdoor-strategy}
A \emph{backdoor  strategy} is a pair $$(\pi^\dagger, \kappa)$$ where 
\begin{itemize}
\item $\pi^\dagger: S \rightarrow \dist(A)$  is a Markov \emph{backdoor policy} that maps current state to a distribution over actions.

\item $\kappa= \langle Q, \mathcal{O}, q_0, \delta, \varkappa \rangle $ is a \emph{finite-memory observation-based trigger},   where:
\begin{itemize}
    \item $Q$ is the finite set of memory states,
    \item $q_0 \in Q$ is the initial memory state,
    \item   $\mathcal{O}  $  is a finite set of observations (inputs). 
    \item $\delta: Q \times \mathcal{O} \rightarrow Q$ is the deterministic transition function that takes a current memory state $q$ and a  new observation $o$ to a next memory state $q' = \delta(q, o)$.
    \item $\varkappa: Q \rightarrow \dist(\mathcal{K})$ is an action selection function which maps the current memory states to a distribution over trigger actions.
\end{itemize}

\end{itemize}
\end{definition}
The reason for choosing finite-memory triggers, also known as finite-state controllers, is that the attacker has only partial observation of the system state. In partially observable stochastic systems, finite-memory policies can achieve superior performance compared to memoryless (Markovian) policies \cite{meuleau2013solving}. 

Given a backdoor strategy \((\pi^\dagger, \kappa)\) and the original \ac{mdp} \(M\), the resulting stochastic process \( M^{(\pi^\dagger, \kappa)} = \{S_t, A_t, O_t\}_{t \ge 0}\) represents the sequence of random variables corresponding to states, actions, and observations induced by the strategy. The following conditions are satisfied:
\[
A_t\sim \pi^\dagger(\cdot |S_t),
\]
\[
O_t \sim E(\cdot |S_t),
\]
\[
S_{t+1}\sim \sum_{k \in \mathcal{K}}  \varkappa(k\mid Q_t) P_k(\cdot|S_t, A_t),
\]
where  $Q_t =\delta(q_0, O_{0:t})$.

 \subsubsection{Attacker's goal}


Using backdoor attack,
the attacker aims to maximize a total discounted return for an adversary reward function  $r_1: S\times A\rightarrow \reals$. 
 However, since the attacker cannot directly control the dynamical system, he employs backdoor attacks    to optimize his attack objective in the perturbed \ac{mdp} $\calM^{\pi^\dagger,\kappa}$ when trigger 
 $\kappa $ is activated and the victim uses the backdoor policy $\pi^\dagger$. When the system is under attack,  the total return to the attacker is  

\begin{align*}
V_1(\mu_0, M^{\pi^\dagger,\kappa}) 
&= \Expect_{{\pi^\dagger,\kappa}} \left[\sum_{t=0}^\infty \gamma^t\, r_1(S_t, A_t) \right],
\end{align*}
and the victim's total return is 
\begin{align*}
V_0 (\mu_0, M^{\pi^\dagger,\kappa}) 
&= \Expect_{{\pi^\dagger,\kappa}} \left[\sum_{t=0}^\infty \gamma^t\, r (S_t, A_t) \right].
\end{align*}
When the initial distribution is clear from the context, we omit it, i.e.,  $ V_i (\mu_0, M^{\pi^\dagger,\kappa}) \coloneqq V_i ( M^{\pi^\dagger,\kappa}) $ for $i \in \{0,1\}$.

A special case of the attacker's objective is when the attacker is antagonistic, i.e.,
\[
r_1(s,a) = -r(s,a), \quad \forall (s,a) \in S \times A.
\]
In this case, the corresponding value functions satisfy
\[
V_1\bigl(M^{\pi^\dagger, \kappa}\bigr) = - V_0\bigl(M^{\pi^\dagger, \kappa}\bigr),
\]because the cumulative reward collected by the attacker is the negation of the defender's reward along any trajectory.

We introduce the stealth constraints in the attack next.


\begin{definition}
	\label{def:stealthy-backdoor}[Stealthy backdoor]
Given the original \ac{mdp} $M = \langle S, A, P, \mu_0,r \rangle $, a  backdoor  strategy $(\pi^\dagger,\kappa)$ is called stealthy if it satisfies two conditions:
\begin{itemize}
\item  The trigger is constrained: At all time $t\ge 0$, for any $(s,a,s')\in S\times A\times S$, it holds that $\abs{P^{(t)}(s'|s,a)  - P_0(s'|s,a)} \le d$,   where $P^{(t)} $ is the transition function determined by the trigger $\kappa$ at time $t$ and $d$ is a small constant.
\item The policy $\pi^\dagger$ is $\varepsilon$-optimal   in the original \ac{mdp} $M$ for a small $\varepsilon$; 
\end{itemize}
\end{definition}

The first condition bounds how much the attacker can perturb the system dynamics at any time. If $d$ is small, then the perturbed dynamics is very close to the nominal dynamics, and thereby the system's trajectory distribution will appear normal and making it hard for the system to detect ongoing attack    during execution with limited data and inherent estimation noise.  The second condition is to ensure the policy $\pi^\dagger$ past the testing phase, as its performance is close to the optimal policy the user can get.

\section{Backdoor Attack Planning in \ac{mdp}s: Methods}
\label{sec:main-results}

First, we show that the first stealthy constraint can be satisfied by choosing the $K$ adversarially perturbed transition functions close to the original transition function $P_0$. The stealthiness of the trigger can be enforced by restricting the adversarial perturbations applied to the system dynamics.


\begin{assumption}
\label{assume:bounded-perturbation}
For any $k\in \{1,\ldots, K\}$, the perturbed transition function $P_k$ is $d$-close to the original transition function $P_0$. That is:
for any $(s,a,s')\in S\times A\times S$, for any $k \in \mathcal{K}$, it holds that 
\[ |{P_k (s'| s,a)-P_0(s'| s,a)}|\le d.
\]
\end{assumption}


With limited data, the hypothesis test cannot reject the null hypothesis that the distribution is nominal when $d$ is small. Under the above assumption, the optimal backdoor strategy planning problem can be formulated as a constrained optimization problem:
\begin{equation}
	\label{eq:opt_backdoor_strategy}
	\begin{aligned}
		\optmax_{\pi^\dagger,\kappa}  \quad &  V_1 (
			M^{\pi^\dagger,\kappa})  \\
			\optst. \quad &   
			 V_0( M^{{\pi^\dagger}})  \geq (1-\varepsilon) 
 V_0^\ast(M),\\
		\end{aligned}    
	\end{equation} where $V_0^\ast(M)$ is the optimal value in the original \ac{mdp} and $V_0(M^{\pi^\dagger})$ is the   value of the backdoor policy $\pi^\dagger$ in the original \ac{mdp}.





We present a planning problem to solve the optimal backdoor strategy within this trigger strategy space and Markov control policy space.  Our method starts with constructing a two-player Markov game  to model the interaction between the agent training the backdoor policy  (referred as P0) and the agent learning how to triggers the backdoor attack (referred as P1).

\begin{definition}[Augmented Constrained Markov Game]
Give a finite-memory observation-based trigger $\kappa= \langle Q, \mathcal{O}, q_0, \delta, \varkappa \rangle $, where $\varkappa$ is to be computed, and the original \ac{mdp} $M$, the following augmented two-player Markov game can be constructed:
\[
\calM = \langle S\times Q, A\times \calK, T, \hat{\mu}_0, R_1 \rangle 
\]
where  
\begin{itemize}
\item $  S\times Q $ is a set of states, which is the cartesian product of the state space in original \ac{mdp} and the memory state space for the trigger; 
\item $  A\times \calK$ is a set of actions, which is the cartesian product of the action space in original \ac{mdp} (P0's action set) and the set of indexes of the transition functions (P1's action set), including perturbed ones and the original one.
\item $T: (S \times Q) \times (A\times \calK) \rightarrow \dist (S\times Q)$ is the transition function, defined by 
\begin{multline*}
	T(  (s',q')| (s,q), (a,k)) =\\ \sum_{o\in \calO} E(o|s') P_k(s'|s,a) \cdot \mathbf{1}( \delta(q,o)=q'),
\end{multline*}
which is the probability of reaching state $s'$ from state $s$ given action $a$ in the transition function $P_k$, while the memory state is updated from $q$ to $q'$ given the observation emitted from that state $s'$.

\item $\hat \mu_0$ is the initial state distribution, defined by  
\[
\hat \mu_0(s,q) = \sum_{o\in \calO} \mu_0(s)E(o|s)\cdot \mathbf{1}(\delta(q_0,o)=q).
\]
  \item $R_1 : (S \times Q) \times (A \times \mathcal{K}) \rightarrow \reals$ is defined by 
$R_1((s, q), (a, k)) = r_1(s, a)$, 
which is the adversarial reward given the state-action pair $(s, a)$.
\end{itemize}
\end{definition}
 
In this Markov game, we consider two players P0 and P1 with restricted policy spaces: P0 uses a memoryless policy over MDP states, \(\pi_0: S \rightarrow \dist(A)\), while P1 uses a Markov policy over finite memory states, \(\pi_1: Q \rightarrow \dist(\mathcal{K})\). The joint policy $\bm{\pi}\coloneqq \pi_0 \otimes \pi_1$, is defined as:
\begin{equation}
	\label{eq:policy-aug}
\bm{\pi}((a,k) \mid (s,q)) = \pi_0(a \mid s) \cdot \pi_1(k \mid q).
\end{equation}
Let \(\Pi_0\) and \(\Pi_1\) be the policy sets of P0 and P1, respectively. The joint policy set is \(\bm{\Pi} = \Pi_0 \otimes \Pi_1\).

\begin{remark}
	Even though that $\pi_1$ does not show any explicit dependency on the   state $s $ in the trigger-perturbed system, it is observed that the automata state $q$ can contain information about $s $. For example, $q$ can be   the most recent $k$  observations.
\end{remark}
The joint policy \(\bm{\pi}\) induces a stochastic process \(\{S_t, Q_t, A_t, K_t\}_{t \ge 0}\) over the augmented Markov game. 
Given $
\calM^{\bm{\pi}}$,
the expected total discounted rewards under the reward function $R_1$ is 
\[
V_1 (
\calM^{\bm{\pi}}) = \Expect_{\bm{\pi}} \left[\sum_{t=0}^\infty \gamma^t  R_1( (S_t,Q_t),(A_t, K_t))  \right].
\]


Finally, we show that the construction of the augmented constrained Markov game constitutes a valid reformulation of the attacker’s stealthy backdoor strategy, and we establish the correspondence between their solutions.

\begin{theorem}
\label{thm:equivalence}
	The solution of a stealthy backdoor strategy in \eqref{eq:opt_backdoor_strategy} is equivalent to solving the following constrained optimization problem:
\begin{equation}
	\label{eq:opt}
	\begin{aligned}
		\optmax_{\bm{\pi}  =\pi_0\otimes \pi_1 \in \bm{\Pi}}  \quad &  V_1 (
			\calM^{\bm{\pi}})  \\
			\optst \quad &  V_0  (
			M^{\pi_0})  \geq (1-\varepsilon) 
 V_0 ^\ast(M),\\
		\end{aligned}    
	\end{equation}
    where $V_0(M^{\pi_0})$ is the value of policy $\pi_0$ in the original MDP $M$ for the original reward $r$ and $V_0^\ast(M) = \max_{\pi_0}V_0 (M^{\pi_0})$ is the optimal value. 

    Let $(\pi_0 , \pi_1 )$ be the solution to \eqref{eq:opt}. A stealthy backdoor strategy is constructed such that:
    \begin{itemize}
        \item The Markov backdoor policy $\pi^\dagger \coloneqq \pi _0$;
        \item The action mapping  $\varkappa: Q\rightarrow \dist(\calK)$ in the finite-memory trigger $\kappa$  is defined such that $\varkappa \coloneqq \pi_1 $.
    \end{itemize} 
\end{theorem}

\section{Computing the backdoor strategy using a blackbox simulator}
\label{sec:algorithm}
The two-player Markov game framework enables the use of the solutions of constrained  games  to solve for the backdoor strategy, which comprises of policies \(\pi_0\) and \(\pi_1\) for P0 and P1, respectively. By this co-design, the backdoor policy  \(\pi_0\) crafted by the attacker is made to be exploitable by the trigger \(\pi_1\) at runtime to maximize the attacker’s return.

So far, our analysis and formulation for solving the backdoor strategy require the model of the original \ac{mdp}.
However, next we show that the attacker does not need access to the system dynamics. As long as the attacker has access to a blackbox simulator of the original system and its observation function, the attacker can still solve for an optimal attack strategy in \eqref{eq:opt}.

We consider parametric policy space for both backdoor policy $\pi_0$ and trigger policy $\pi_1$, which are parametrized by vectors $\theta_0$ and $\theta_1$, respectively. Using the attacker's return as the shared objective between P0 and P1, the existing algorithm in  \cite{jordanIndependentLearningConstrained2024} can be applied to solve a constrained Nash equilibrium of the constrained Markov game.  These methods are based on the recent progress  \cite{boobStochasticFirstorderMethods2023,jiaFirstOrderMethodsNonsmooth2025} for solving nonconvex  optimization
problems with nonconvex functional constraints. 


Algorithm~\ref{alg:switching} outlines the overall solution procedure. For clarity, let \(\theta_i^{(t)}\) denote the parameters of policy \(\pi_i^{(t)}\) at iteration \(t\). The quantity \(V_0(\theta_0^{(t)}, M)\) represents the total discounted reward of policy \(\pi_0^{(t)}\) evaluated in the original \ac{mdp} \(M\), while \(V_1(\theta_0^{(t)}, \theta_1^{(t)}, \mathcal{M})\) denotes the corresponding evaluation in the augmented Markov game \(\mathcal{M}\) with respect to the attacker’s reward \(R_1\). The algorithm finds a local optimum to the optimization problem~\eqref{eq:opt}, which has non-convex objective and constraint functions. 

\begin{algorithm}
    \caption{Computing a backdoor strategy $\pi_0, \pi_1$. }\label{alg:switching}
    \begin{algorithmic}
        \STATE Initialize $\theta_0^{(0)}$, $\theta_1^{(0)}$ (initial policy parameters), $\alpha_0$, $\beta_0$ (learning rates), $\varepsilon$ (sub-optimality threshold), $\alpha$ (stopping threshold).
        \FOR{$t=1$ \TO $T$}  
        \STATE $b \gets V^\ast_0(M)(1 - \varepsilon)$ \COMMENT{$V^\ast_0(M)$ is obtained by applying RL to the blackbox simulator.}
          \STATE Sample $m$ trajectories $\{\tau^{(i)}\}_{i=1}^m$ with the simulator in augmented $\calM^{\bm{\pi}}$ based on $\pi_{\theta_0^{(t)}}$ and $\pi_{\theta_1^{(t)}}$;
          \STATE  Sample $m$ trajectories $\{\rho^{(i)}\}_{i=1}^m$ from  original $\ac{mdp}$ based on $\pi_{\theta_0^{(t)}}$.
            \IF{$V_0( \theta_0^{(t)}, M) < b $ \COMMENT{Evaluate with $\{\rho^{(i)}\}_{i=1}^m $.}}  
            \STATE $\theta_0^{(t+1)} = \theta_0^{(t)} + \alpha_t \hat{\nabla}_{\theta_0} V_0( \theta_0^{(t)}, M)$  \COMMENT{Approximate gradient using $\{\rho^{(i)}\}_{i=1}^m$}
               \STATE $\theta_1^{(t+1)}  = \theta_1^{(t)}$;
            \ELSE 
            \STATE $\theta_0^{(t+1)} = \theta_0^{(t)} + \alpha_t \hat{\nabla}_{\theta_0} V_1(  \theta_0^{(t)}, \theta_1^{(t)} \calM )$
             \STATE $\theta_{1}^{(t+1)} =  \theta_{1}^{(t)} + \beta_t \hat{\nabla}_{\theta_1} V_1( \theta_0^{(t)}, \theta_1^{(t)}  , \calM )$, \COMMENT{Approximate gradient using $\{\tau^{(i)}\}_{i=1}^m$}
            \ENDIF
            \IF{$ \|\theta_0^{(t+1)} - \theta_0^{(t)}\| \leq \alpha$ and $\|\theta_1^{(t+1)} - \theta_1^{(t)}\| \leq \alpha $}
            \STATE \textbf{return} ($\theta_0^{(t)}, \theta_0^{(t)}$)
            \ENDIF
        \ENDFOR  
        \RETURN ($\theta_0^{(T)}, \theta_0^{(T)}$)
    \end{algorithmic}
\end{algorithm}

At each iteration, the algorithm samples \(m\) trajectories from a black-box simulator of the system and generates corresponding observations using the emission function, yielding \(m\) trajectories of observations. It then estimates the total return of \(\pi_0^{(t)}\) from these samples and checks whether the constraint in ~\eqref{eq:opt} is satisfied. The update strategy is determined by whether the constraint is satisfied:
\begin{itemize}
    \item 
   If the constraint is satisfied, both \(\theta_0^{(t)}\) and \(\theta_1^{(t)}\) are updated via a policy gradient step in the augmented game \(\mathcal{M}\), optimizing the attacker's objective.
\item If the constraint is violated, only \(\theta_0^{(t)}\) is updated using a policy gradient step in the original MDP \(M\) to improve performance with respect to the original reward, while \(\theta_1^{(t)}\) remains fixed.
\end{itemize}

This alternating (or switching) policy gradient process continues until convergence or a maximum number of iterations is reached.

The policy gradient update is based on the REINFORCE estimator:
\begin{equation}
    \hat{\nabla}_{\theta} V(\theta) = \frac{1}{m} \sum_{i=1}^m
    \nabla_\theta \log\left(
      \prod_{t=1}^{|\tau^{(i)}|} \pi_\theta\left(a_t^{(i)} \mid o_t^{(i)}\right)
    \right) R\left(\tau^{(i)}\right),
\end{equation}
where $R\left(\tau^{(i)}\right)$ is the total return with the sampled trajectory  $\tau^{(i)}$.

Note that \(\hat{\nabla}_{\theta_0} V_1(\theta_0^{(t)}, \theta_1^{(t)}, \mathcal{M})\) denotes the partial derivative of \(V_1\) with respect to \(\theta_0\), and similarly, \(\hat{\nabla}_{\theta_1} V_1(\theta_0^{(t)}, \theta_1^{(t)}, \mathcal{M})\) is the partial derivative with respect to \(\theta_1\).

\section{Experimental Results}
 \label{sec:experiment}

We consider a robot motion planning problem in a stochastic Gridworld shown in Fig. ~\ref{fig:small-Gridworld}. The robot can move in four compass directions. Given an action, say, ``N'', the robot enters the intended cell with $1 - 2\alpha$ probability and enters the neighboring cells, which are west and east cells with $\alpha$ probability. 
The environment contains two types of targets. Reaching a high-value target (marked by a red triangular flag) yields a reward of 20, while reaching a low-value target (marked by a yellow rectangular flag) yields a reward of 2. Additionally, the environment includes traps, which incur a penalty of $–10$ when reached. When a target or a trap is reached, the robot stays in the same cell without any future reward. The agent receives a reward of 0 at all other states. The robot's objective is to maximize the total discounted rewards, under the discounting factor $\gamma =0.99$.

\begin{figure}[H]
    \centering \includegraphics[width=0.6\linewidth]{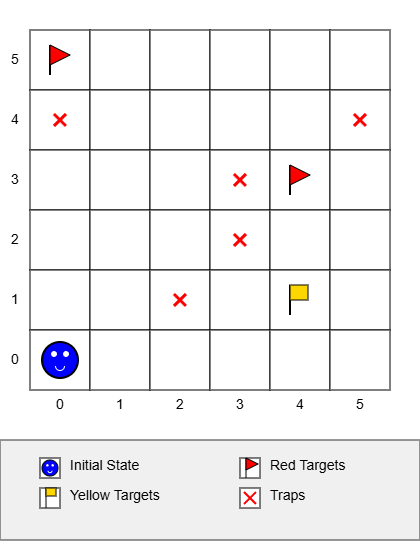}
    \caption{\small A $6 \times 6$ Stochastic Gridworld}
    \vspace{-2ex}
    \label{fig:small-Gridworld}
\end{figure}

The original, attack-free dynamics is characterized by the stochasticity parameter $\alpha=0.1$. 
In this experiment, we consider the case where the attacker can perturb the system using any stochasticity parameter in $\{0, 0.3\}$, which corresponds to $d = 0.2$ in Assumption~\ref{assume:bounded-perturbation}. For the P1's action set $\mathcal{K}$, we construct two transition functions $P_1, P_2$ with stochasticity parameters $0$ and $0.3$, respectively.

The attacker’s partial observation is defined so that, at each state $s$, the attacker observes the state with probability $p_{obs} = 0.8$, and receives an empty observation with probability $1 - p_{obs}$. We consider a zero-sum case where $r_1(s,a) = - r(s,a)$ for $(s,a)\in S\times A$. 

A trigger policy with a finite-memory of size $m$ is used. Given an integer $m\ge 1$, the state set $Q=\{\mathcal{O}^{\le m}\}$ are the set of observation sequences with length $\le m$.  For each $q\in Q$, $\delta(q, o) = q'$ is defined such that $q' = \mathsf{suffix}^{= m}(q\cdot o)$ \footnote{$\mathsf{suffix}^{= m}(w)$ is the last $m$ symbols of string $w$ if $|w|\ge m$ or $w$ itself otherwise.} is the last  (up to) $ m$    states. The   probabilistic output function is parameterized as, for integer $1\le k\le K $,
$     \varkappa_\theta(k|q) =  \frac{\exp(\theta_{q,k})}{\sum_{ k'}\exp(\theta_{q, k'})},
$ where  $\theta \in \reals^{
|Q \times K|}$ is the trigger parameter vector.  With the transition function defined, the trigger with $\theta_1$- parameterized output $\varkappa$ is the trigger policy $\pi_1$. The Markov policy $\pi_0$ also use a softmax parameterization, \ie, for any $(s,a)\in S\times A$, $\pi_{\theta }(a|s) =  \frac{\exp(\theta_{s,a})}{\sum_{ a'}\exp(\theta_{s, a'})}$. The policy with parameter $\theta_0$ is the backdoor policy $\pi_0  $, which is also $\pi^\dagger$.

We consider three performance metrics: \begin{inparaenum}[1)]
\item $V_0(\theta_0, M)$, which is the value of the  backdoor policy $\pi^\dagger$ in the attack-free case, evaluated with the robot's reward function $r$ and the original \ac{mdp}.
\item $V_0(\theta_0, \theta_1, \calM)$, which is the value of the  backdoor policy $\pi^\dagger$ when the trigger is active, evaluated with the robot's reward function $r$.
\item $V_1(\theta_0, \theta_1, \calM)$, which is the value of the backdoor policy $\pi^\dagger$ 
when the trigger is active, evaluated with the attacker's reward function $r_1$.  
\end{inparaenum}
In the zero-sum case, $V_0(\theta_0, \theta_1, \calM) = -V_1(\theta_0, \theta_1, \calM)$; therefore, we omit the plot of $V_1(\theta_0, \theta_1, \calM)$ as it is redundant. \footnote{The code is available at \url{https://github.com/LeahWeii/mdp-backdoor-synthesis-clean}}

We set the sub-optimality parameter for the backdoor policy to $\varepsilon = 0.2$, Fig.~\ref{fig:convergence-gw} presents convergence trends of the $V_0(\theta^{(t)}_0, \mathcal{M})$, $V_0(\theta_0, \theta_1, \mathcal{M})$ over the iterations of Algorithm~\ref{alg:switching}, where the algorithm converges to a local optimum. We observe that when the system is not under attack, the backdoor policy achieves a value of $13.49$. While when the trigger is active, its performance degrades to only $1.27$.

To illustrate how the near-optimality parameter \( \varepsilon \) in the constraint \( V_0(\theta_0, M) \ge (1 - \varepsilon) V_0^\ast(M) \) affects the attacker's performance,
 Fig.~\ref{fig:convergence-gw-epsilon} plots the converged values of $V_0(\theta^{\ast}_0, \mathcal{M})$, $V_0(\theta^\ast_0, \theta^\ast_1, \mathcal{M})$ given different values of $\varepsilon$ under the same P1's action set $\mathcal{K}$. When $\varepsilon=0$, the backdoor policy $\pi^\dagger$ achieves the optimal performance without trigger, and when the trigger is activated the performance degrades to $7.67$ . A larger $\varepsilon$ leads to improved performance, as it expands the feasible space of the policy $\pi^\dagger$ in \eqref{eq:opt_backdoor_strategy}.
 \begin{figure}[H]
    \centering
    \includegraphics[width=0.75\columnwidth]{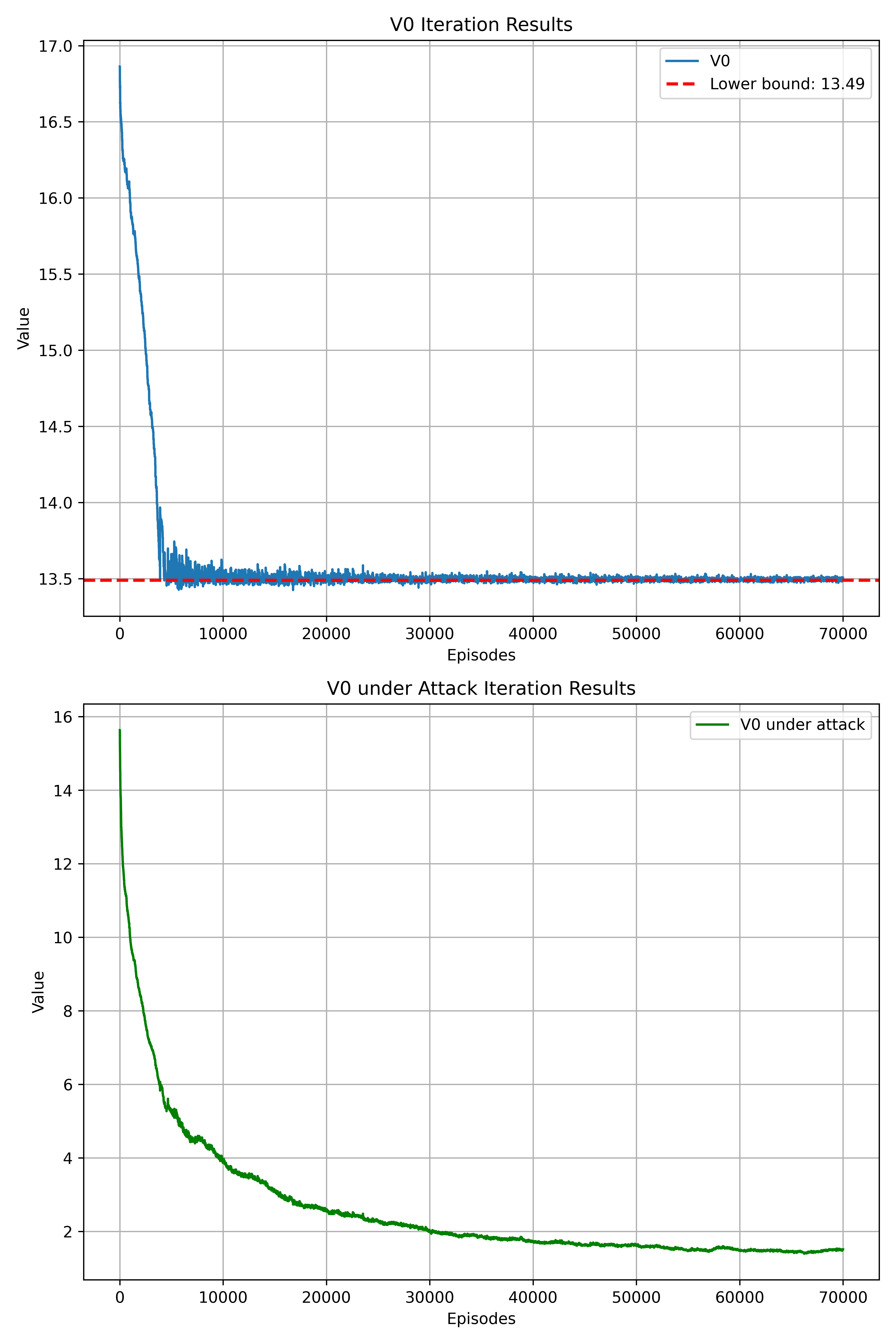}
    \caption{The evaluations of policies over iterations: $V_0(\theta_0^t, M)$ and $V_0(\theta_0^t, \theta_1^t, \mathcal{M})$ in the zero-sum case.}
    \label{fig:convergence-gw}
\end{figure}

To illustrate how the P1's action set $\mathcal{K}$, which is influenced by the rectangular uncertainty parameter $d$, affects the attacker's performance, Fig.~\ref{fig:convergence-gw-delta} presents the values of $V_0(\theta^\ast_0, M)$ and $V_0(\theta^\ast_0, \theta^\ast_1, \mathcal{M})$ under different perturbed stochasticity parameters $\delta$.
In this set of experiments, we allow the attacker to perturb the system using any stochasticity parameters in $\{ \max(\alpha-\delta, 0), \alpha+\delta\}$ for $d \in \{0, 0.05,0.1,0.2, 0.3\}$ under the same sub-optimality parameter $\varepsilon=0.2$.

 \begin{figure*}[ht]
    \centering
    \begin{subfigure}[b]{0.4\textwidth}
        \centering
        \includegraphics[width=\textwidth]{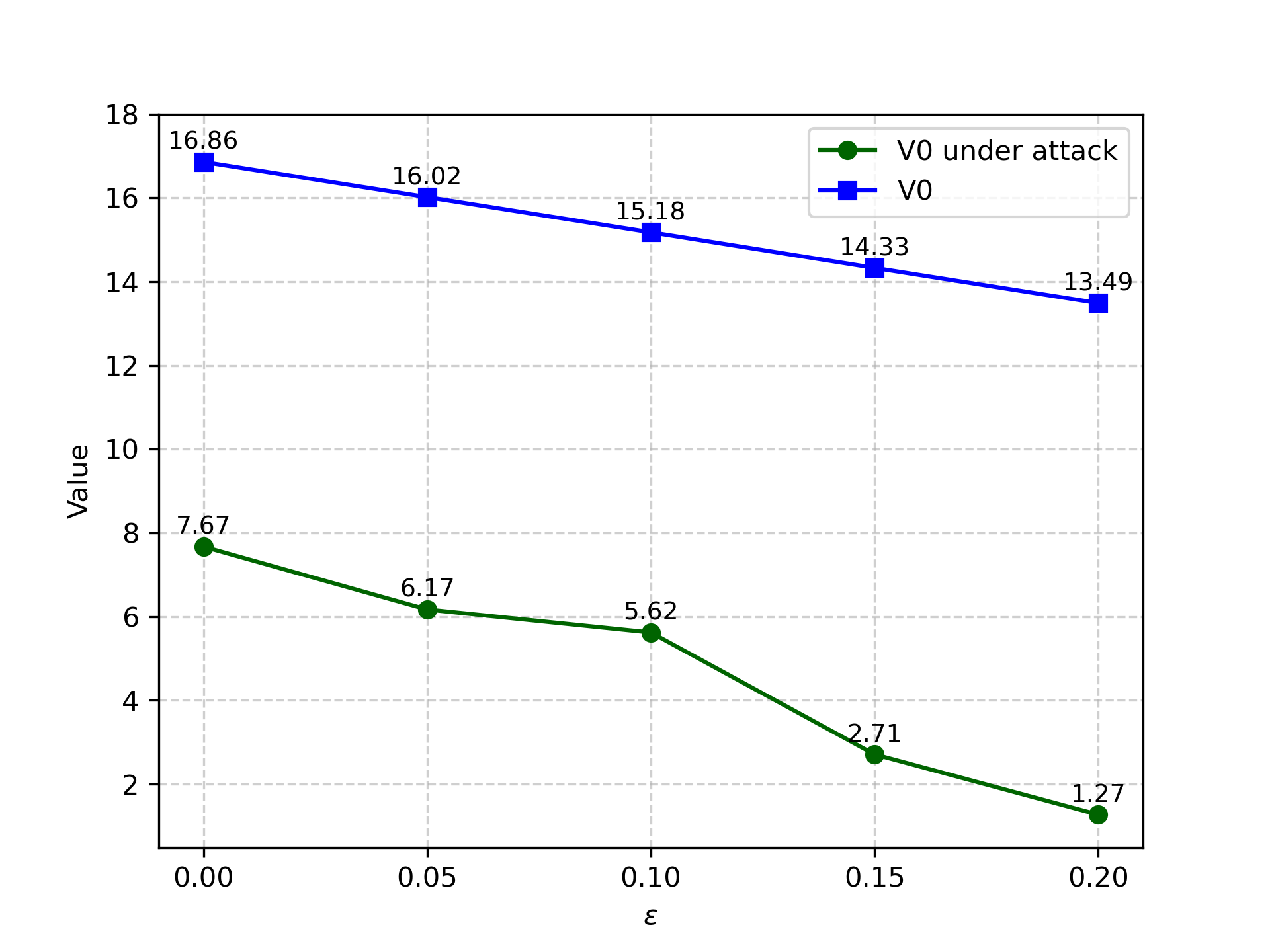}
        \caption{Outputs of $V_0(\theta_0^\ast, M)$ and $V_0(\theta_0^\ast, \theta_1^\ast, \mathcal{M})$ under different values of $\varepsilon$.}
        \label{fig:convergence-gw-epsilon}
    \end{subfigure}
    \hfill
    \begin{subfigure}[b]{0.4\textwidth}
        \centering
        \includegraphics[width=\textwidth]{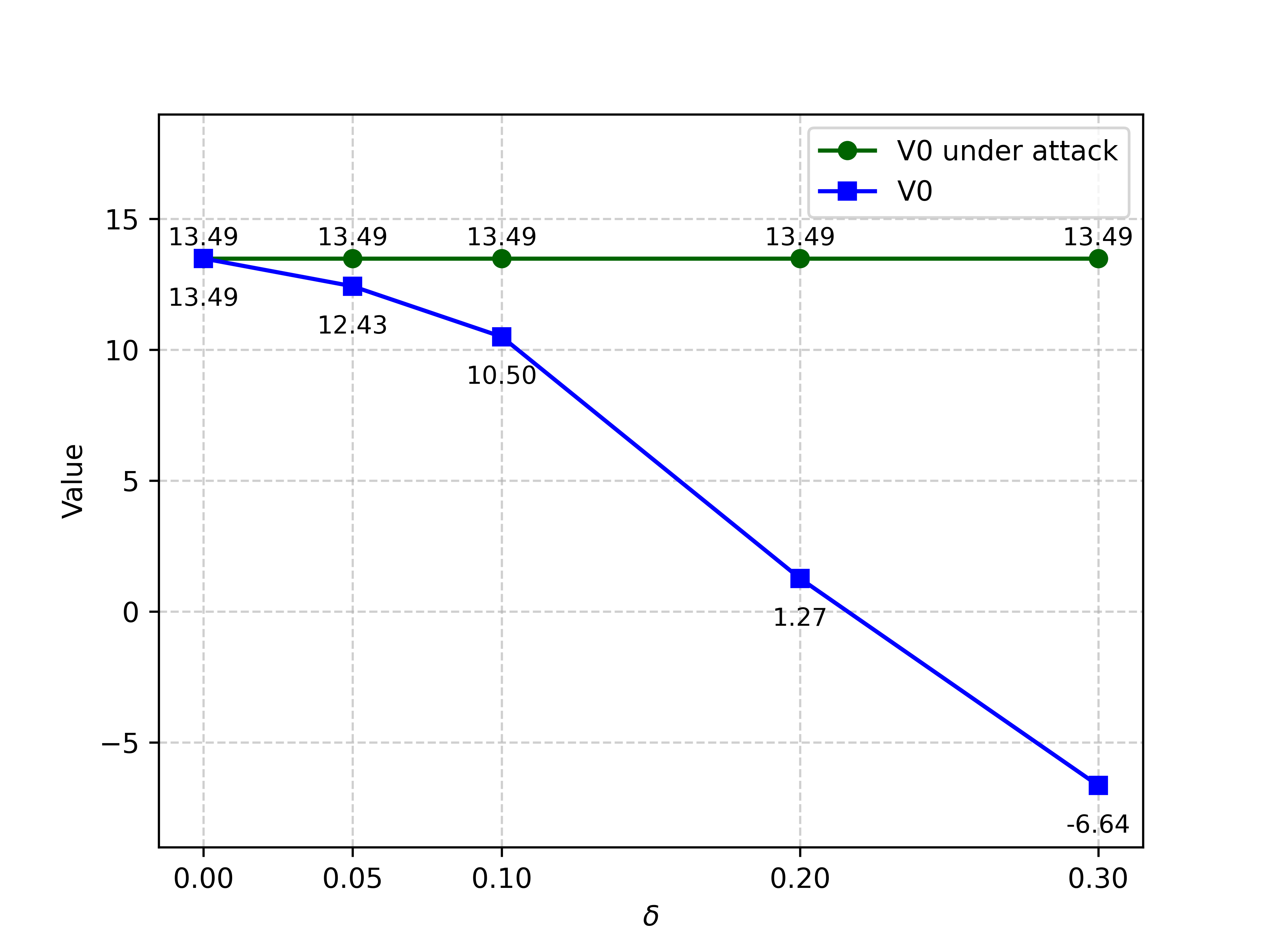}
        \caption{Outputs of $V_0(\theta_0^\ast, M)$ and $V_0(\theta_0^\ast, \theta_1^\ast, \mathcal{M})$ under different values of $\delta$.}
        \label{fig:convergence-gw-delta}
    \end{subfigure}
    \caption{Comparison of attack performances under different parameters $\varepsilon$ and $\delta$.}
    \label{fig:convergence-gw-both}
\end{figure*}

We also consider another general case, where the reward of the attacker and victim is not zero-sum. The victim’s reward function remains unchanged. The attacker’s reward is defined as $-2$ for reaching a red triangular flag, $20$ for reaching a yellow rectangular flag, and $10$ for reaching a trap.Fig.~\ref{fig:convergence-gw-not-zero-sum} shows the convergence of $V_0(\theta^{(t)}_0,M)$ and $V_0(\theta_0^{(t)},\theta_1^{(t)},\mathcal{M})$ across iterations of Algorithm~\ref{alg:switching} for $\varepsilon=0.2$. When the trigger is activated, the target agent’s performance falls to $-1.40$, while the attacker attains a value of $4.73$.  This demonstrates an effective backdoor: activation of the trigger produces a substantial performance drop compared with the case in which the backdoor policy is present but the trigger is not activated.

It is worth noting that the zero-sum reward design is not necessarily the most challenging competitive scenario. For instance, if   the trap reward for the attacker is increased to $100$ compared to the zero-sum case, the setup  becomes non–zero-sum. The attacker's gains are amplified when the victim reaches a trap state---as a result, the backdoor attack introduces a higher probability for the victim to reach a trap. In this environment, because the attacker assigns a higher reward to the yellow target, the coordinated backdoor and trigger policies   tend to increase the likelihood of trajectories near the yellow flag  and consequentially raise the likelihood of falling into nearby traps, at which the victim receives $-10$. Meanwhile, the attacker’s improved performance can be partly attributed to the higher reward design. 

\begin{figure}[H]
    \centering
    \includegraphics[width=0.7\columnwidth]{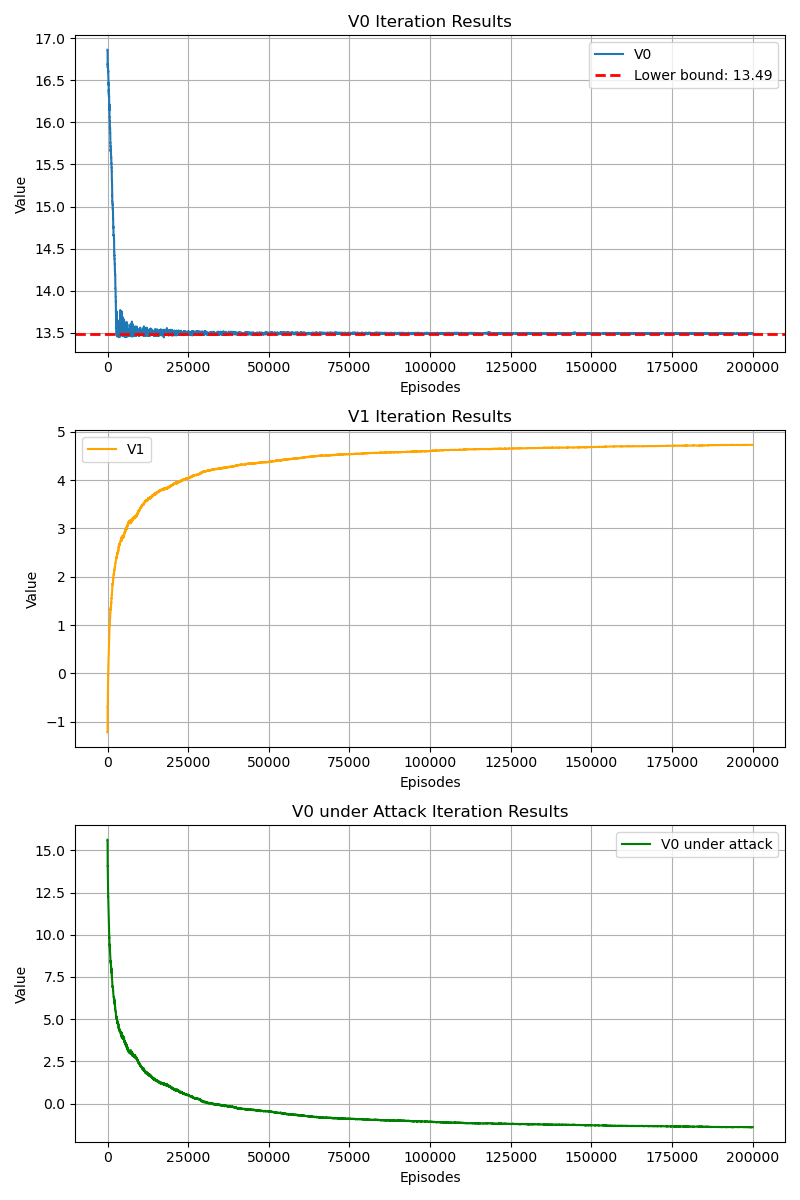}
    \caption{The evaluations of policies over iterations: $V_0(\theta_0^t, M)$ and $V_0(\theta_0^t, \theta_1^t, \mathcal{M})$ in the non-zero-sum case.}
    \label{fig:convergence-gw-not-zero-sum}
\end{figure}

\section{Conclusion}
This paper presents a method to design backdoor attack strategies targeted as probabilistic sequential decision-making systems. It shows the co-design of backdoor policy and trigger can be formulated as an augmented constrained Markov game with a shared constraint. Building on this game-theoretic modeling, it employs two-agent independent policy gradients with shared constraints to solve the backdoor strategy. We demonstrate the effectiveness of the co‑design on a  stochastic gridworld.

Future work can be extended in multiple directions. First, backdoor attacks in MDPs provide critical insights for designing backdoor attacks in model-free reinforcement learning. This is because the backdoor policy can be used to compute data poisoning attacks in offline RL.  
Second,  the adverse effect of a backdoor attack highly depends on the environment configuration. This observation could lead to insights for  developing security defenses against such attacks and to construct an attack-aware verification or   planning/reinforcement learning algorithms that are robust to adversarial perturbations. 

\section*{Acknowledgment}
Research was sponsored by Army Research Office  under Grant
Number W911NF-22-1-0166, and NSF under award \#2144113. The views and conclusions contained in this document are those of the authors and
should not be interpreted as representing the official policies, either expressed or implied, of the Army Research
Office or the U.S. Government.

\addtolength{\textheight}{-12cm}   





\bibliographystyle{ieeetr}
\bibliography{refs}

\end{document}